\documentclass[runningheads]{llncs}

\usepackage[T1]{fontenc}
\usepackage{amsmath,amssymb,amsfonts}
\usepackage{graphicx}
\usepackage{xcolor}
\usepackage{tikz}
\usepackage{pgfplots}
\usepackage{cleveref}

\pgfplotsset{compat=1.18}
\usetikzlibrary{matrix,fit}

\newcommand{\avg}[1]{\bar{#1}}
\newcommand{\avp}{\mbox{$\avg{p}$}}
\newcommand{\ost}{{\omega_*}}
\newcommand{\odg}{{\omega_\dag}}
\newcommand{\sfrac}[2]{{\textstyle\frac{#1}{#2}}}
\newcommand{\A}{\mathcal{A}}
\newcommand{\W}{\mathcal{W}}
\newcommand{\Cov}[2]{\mathrm{Cov}\!\left(#1,#2\right)}

\begin{document}

\title{Bridging Dependence: From Opinion Leader
Influence to Standard Correlation Measures}

\titlerunning{Bridging Dependence}

\author{Jonas Karge)}

\authorrunning{Jonas Karge}

\institute{TU Dresden}

\maketitle              
\begin{abstract}
As a cornerstone result of social choice and democracy theory, the Condorcet Jury Theorem (CJT) provides probabilistic guarantees for the correctness of majority decisions under idealized conditions, in particular assuming that voters are statistically independent.
When generalizing the original CJT, a standard approach in the literature to model dependence introduces an external \textit{opinion leader} (OL) whose influence affects individual voting behavior. In this paper, we take the first steps toward a unified framework for modeling dependence by translating the OL model into standard correlation measures, namely covariance and correlation coefficient. These translations build conceptual bridges between different dependence models in the voting literature and facilitate the application of existing theoretical results across frameworks.

\end{abstract}

\section{Introduction}

As a cornerstone result of social choice and democracy theory, the Condorcet Jury Theorem (CJT) provides probabilistic guarantees for the correctness of majority decisions under idealized conditions, in particular assuming that voters are statistically independent.
When generalizing the original CJT, a standard approach in the literature to model dependence introduces an external \textit{opinion leader} (OL) whose influence affects individual voting behavior. In this paper, we take the first steps toward a unified framework for modeling dependence by translating the OL model into standard correlation measures, namely covariance and correlation coefficient. These translations build conceptual bridges between different dependence models in the voting literature and facilitate the application of existing theoretical results across frameworks.

The CJT has found widespread interest across disciplines. 
For example, recent advances in artificial intelligence have sparked growing interest in social choice theory, the study of how individual preferences or judgments can be aggregated into collective decisions. Concepts from social choice theory are particularly relevant to AI in domains requiring opinion aggregation, such as voting systems, federated learning, or multi-agent coordination, where groups must reconcile diverse inputs to reach coherent outcomes. When aggregating opinions through voting, two primary objectives emerge: ensuring procedurally fair processes or identifying rules that reliably track the correct alternative (i.e., an underlying ground truth) \cite{dietrich2021social}. In this work, we adopt the latter perspective, aligning with the tradition of epistemic voting, which evaluates voting rules based on their ability to discern truth. At the core of epistemic voting lies the \textit{Condorcet Jury Theorem}, which offers probabilistic guarantees that the correct alternative will be selected through voting, provided certain assumptions hold.

The original formulation of the Condorcet Jury Theorem (CJT) is based on several key assumptions: agents are equally competent \textit{(homogeneity)}, they are more likely to vote for the correct alternative than for an incorrect one \textit{(reliability)}, they act independently, uninfluenced by each other or by any external factors \textit{(independence)}, and they select exactly one option \textit{(completeness)} from a set of precisely two alternatives \textit{(dichotomy)} under majority rule \cite{KR2022-21}. Given these conditions, the classical CJT \cite{Condorcet} establishes the following result.

\begin{theorem}
For odd-num\-be\-red, homogeneous groups of independent and reliable agents in a dichotomic voting setting, the probability that majority voting identifies the correct alternative 
\begin{itemize}
\item increases monotonically with the number of agents and 
\item converges to $1$ as the number of agents goes to infinity.
\end{itemize}
\end{theorem}

\paragraph{Condorcet Jury Theorem in AI safety, robustness, and ethics.}
The CJT underpins many AI applications where a \emph{wisdom-of-crowds} effect can improve accuracy through aggregation.  For \emph{risk quantification} and \emph{robustness}, recent work leverages CJT-style majority rules together with Byzantine-resilient consensus, proposing \emph{consensus learning} to secure collective model decisions in distributed settings \cite{magureanu2024consensus}, akin to machine learning models where \emph{ensemble} methods aggregate diverse classifiers via voting and typically outperform single models, leveraging CJT-like guarantees \cite{Dietterich,Lam}. For example, in safety-critical perception, CJT principles have been used to derive voting-based ensemble scores that improve efficiency in COVID-19 identification \cite{srivastava2022ensemble} and enhance detection of colon cancer in histopathology \cite{srivastava2023cjt}. 

Beyond risk analysis, the CJT has entered \emph{AI safety and ethics} debates. For example, it has been examined whether LLMs can strengthen democratic practice (e.g., by summarizing public consultations or assisting policy deliberation) and assess the extent to which such tools preserve the epistemic virtues that make democratic aggregation attractive in the first place \cite{lazar2024can}. Relatedly, it has been analyzed whether AI trained on \emph{crowdsourced moral judgments} can achieve “moral expertise’’ via wisdom-of-crowds mechanisms, while highlighting legitimacy challenges under reasonable moral disagreement \cite{schuster2025moral}. Finally, it has been evaluated whether \emph{AI alignment} should be determined democratically by affected stakeholders or deferred to normative experts, discussing the CJT as a potential justification and its limits when independence, competence, or shared ground truth are in doubt \cite{steingruber2025justifications}. 


Despite its wide applicability, the CJT relies on idealized assumptions that rarely hold in practice. A central research goal has therefore been to generalize the theorem while preserving its asymptotic conclusions under weaker conditions. Yet, once heterogeneity in agent competence is introduced, the monotonicity result no longer holds \cite{Owen}.
In particular, the CJT’s guarantees rely on \emph{independence} of voters and better-than-chance \emph{competence}. Correlated errors and systematic biases can erode or reverse the majority-vote advantage. Therefore, a central concern in order to make safety claims with CJT-style ensembles is to address the independence assumption, and to make correlation quantifiable. 
 
\paragraph{Weakening Independence.} A large body of work has focused on relaxing the independence assumption by allowing interdependencies among voters. A well-studied approach introduces an \textit{opinion leader} (OL), representing an external influence. The OL does not vote directly but evaluates alternatives with competence~$\hat{p}$ (the probability of choosing the correct option). Each agent then follows the OL’s choice with probability~$\pi$ and otherwise relies on their own judgment \cite{OL}.

In one of the earliest generalizations of the CJT to include such influence, Boland, Proschan, and Tong assume homogeneous agent competence levels $p$ and an equally competent OL, i.e., $\hat{p} = p$, within a dichotomous voting framework using majority rule \cite{Boland} . Berg  subsequently extended this model to encompass weighted voting rules, still within a dichotomous setting \cite{Berg}. Later, Goodin and Spiekermann  relaxed the assumption of equal competence between agents and the OL, allowing $\hat{p} \neq p$ \cite{Goodin}. They also identified a precise threshold at which the asymptotic success of the CJT breaks down, depending on both $\pi$ and $p$.
Beyond the OL model, several other approaches have explored voter dependence in CJT generalizations. Ladha incorporates pairwise correlations among voters, constrained by average correlation coefficients \cite{Ladha}. Kaniovski characterizes classes of joint distributions where the overall jury competence may improve or deteriorate, depending on voter competence and correlation structure \cite{Kaniovski}. Pivato likewise permits dependence but restricts the average covariance among voters \cite{Pivato}.

Despite the breadth of approaches, establishing models of voter dependence that are simultaneously rigorous, transparent, and practically applicable remains an ongoing challenge in jury theorem research \cite{sep-jury-theorems}. Addressing this gap, our work takes initial steps toward unifying different notions of dependence under a common framework.

\paragraph{Our Contribution.}
We provide the first explicit translation of OL-induced dependence into standard statistical metrics, namely, covariance and correlation, thereby unifying different dependence models in jury theorems. Our objectives are twofold:  
\begin{enumerate}
    \item to establish a conceptual bridge that enables theoretical results to transfer across existing frameworks, and  
    \item to enhance practical applicability by expressing OL influence in terms of familiar correlation measures.  
\end{enumerate}


To this end, we first present the underlying voting and probabilistic framework, as well as a recently proposed jury theorem by \cite{OL}, which incorporates the OL influence while weakening all other original assumptions and providing probabilistic guarantees for identifying the correct alternative. We then derive translations from the OL model to standard correlation measures. In addition, we show how existing bounds, specifically, the result of \cite{OL}, can be applied to these translated correlation metrics.

\section{Preliminaries}

The \textit{Condorcet Jury Theorem} (CJT) is a cornerstone of voting theory, providing probabilistic guarantees for identifying the correct alternative under a set of idealized assumptions. A recent generalization by \cite{OL} relaxes all of these assumptions simultaneously. In this extended framework, agents may vote over any finite number of alternatives, individual competence levels may vary, and correlation among voters is introduced through the classical \textit{opinion leader} (OL) model \cite{Boland}.


Before presenting this result, we first introduce the underlying voting and probabilistic framework, following the exposition of Karge et al. \cite{OL}.

\paragraph{The Formal Voting Framework.} As the basic building blocks of our voting framework, let $\,\W = \{\omega_1,\ldots, \omega_m\}$ denote a finite set of $m$ alternatives, and let $\A = \{a_1,\ldots, a_n\}$ denote a finite set of agents. With this, a single \textit{approval voting (instance)} can be represented by the relation $V \subseteq \A \times \W$, where $(a_{i},\omega_j)\in V$ means that agent $a_{i}$ approves choice $\omega_j$. Then we define the \textit{score} 
$\#_{\!V}\,\omega$ of a choice $\omega \in \W$ as $\#_{\!V}\,\omega = |\{ a_{i} \in \A \mid (a_{i},\omega) \in V \}|.$ Finally, the winning alternative is defined as the alternative that gets a strictly higher score than any other alternative: $\#_{\!V}\,\omega > \max_{\omega'\in \W \setminus \{\omega\}} \#_{\!V}\,\omega'$.

\paragraph{Formal Probabilistic Model.}
The voting scenario described is modeled by a random process that generates the correct alternative, $\ost$, the OL's choice, as well as the agent's \textit{privately approved choice}: this is the choice the agent would make if the agent were not influenced by the OL. This process is governed by a joint probability distribution $\mathbb{P}$ over Bernoulli (i.e., $\{0,1\}$-valued) random variables, which can be represented as follows:
    
$$
\begin{array}{c@{\ }c@{\ }c}
X^{\omega_1}_*, &\ldots &, X^{\omega_m}_*,\\ 
X^{\omega_1}_{o}, &\ldots &, X^{\omega_m}_{o},\\ 
X^{\omega_1}_{1}, &\ldots &, X^{\omega_m}_{1},\\ 
\vdots & \vdots\,\vdots\,\vdots & \vdots\\  
X^{\omega_1}_{n}, &\ldots &, X^{\omega_m}_{n}.\\ 
\end{array}
$$
The values taken by these random variables represent the outcome of a voting event as follows: 

\begin{itemize}
\item $X^{\omega_j}_*$ is $1$ if $\omega_j$ is the true world state (i.e, $\omega_j = \ost$), otherwise $0$, 
\item $X^{\omega_j}_o$ is $1$ if the OL approves $\omega_j$, otherwise $0$, 
\item $X^{\omega_j}_{i}$ represents the $i$th agent's private signal regarding his approval of the $j$th world state: it is $1$ if $a_i$ privately approves $\omega_j$ and 0 otherwise. 
\end{itemize}

Given the joint distribution, we define the random variable~$V^{\omega_j}_i$ to represent the final outcome of agent~$i$'s vote in world state~$\omega_j$, i.e., after potential influence by the opinion leader. By assumption, this vote is a probabilistic mixture: with probability~$\pi$, agent~$i$ adopts the OL's vote~$X^{\omega_j}_o$, and with probability~$1 - \pi$, they follow their own judgment~$X^{\omega_j}_i$. Accordingly, for any $x \in \{0,1\}$, we have:
$$
\mathbb{P}(V^{\omega_j}_i = x) = \pi \, \mathbb{P}(X^{\omega_j}_o = x) + (1 - \pi) \, \mathbb{P}(X^{\omega_j}_i = x).
$$

Let $p^{\omega}_1, \ldots, p^{\omega}_n$ denote the Bernoulli parameters of the agents’ individual (i.e., ``inner voice'') random variables $X^{\omega}_1, \ldots, X^{\omega}_n$, for each $\omega \in \W$. That is, $p^{\omega_j}_i = \mathbb{P}(X^{\omega_j}_i = 1)$. Similarly, for each $\omega \in \W$, let $\hat{p}^{\omega_1}, \ldots, \hat{p}^{\omega_m}$ be the Bernoulli parameters of the OL's vote variables $X^{\omega_1}_o, \ldots, X^{\omega_m}_o$.

The probability that the OL approves the correct alternative (i.e., that $X^{\ost}_o = 1$) is given by $\hat{p} = \mathbb{P}(X^{\omega_*}_o = 1)$. For convenience, we abbreviate the event that $\omega_j$ is the correct world state as $[\omega_* = \omega_j]$, which remains unknown to the agents.

We now come to two central assumptions about the joint distribution that are necessary to derive the asymptotic result. The first assumption states that, although the agents are correlated through the OL's choice, we assume their private signals to be independent.
Conditioning upon the actual state of the world, we can define \textit{private agent approval independence} as follows:

\begin{definition}\label{def:independence}
Conditioned on the true state, private approvals are independent across agents: for every $\omega \in \{\omega_1,\ldots,\omega_m\}$ and every $v_1,\ldots,v_n \in \{0,1\}$,
\[
\mathbb{P}\!\left(\bigwedge_{i=1}^n X_i^{\omega}=v_i \,\middle|\, \ost=\omega\right)
\;=\;
\prod_{i=1}^n \mathbb{P}\!\left(X_i^{\omega}=v_i \,\middle|\, \ost=\omega\right).
\]
\end{definition}

That is, conditioned on the true world state, the joint probability of any given pattern of
private approval decisions with respect to a given alternative can be computed by taking the
product of the corresponding marginal probabilities.


As a second key assumption, we deal with the competencies of agents with respect to their privately approved choices, governed by the parameter $p^{\omega}_{k}$ of the $k$th agent with respect to its ability to identify the true world state among any number of alternatives when uninfluenced by the OL. We denote the average of these competencies by $$\quad \avp^{\omega}= \frac{1}{n} \sum\limits_{k=1}^n p^{\omega}_{k}$$ and formalize the second assumption as follows \cite{OL}:

\begin{definition}\label{def:reliability}
A joint probability distribution satisfies \textit{$\Delta p$-group reliability} for some $\Delta p > 0$, if the probability, with respect to the agent's inner voice, to approve the true world state, averaged across all agents, is at least $\Delta p$ higher than the averaged probability for approving any other state, i.e., for 
every $n$ and $\omega_\dag\in \W \setminus \{\ost\}$ holds
$$\avp^{\ost} \geq \Delta p + \avp^{\odg}.$$
\end{definition}

From these assumptions, a bound on a minimum success probability, $P_{\text{min}}$, for the electorate to identify the correct alternative can be derived \cite{OL}. 
Observe that this bound converges to 1 as the number of agents goes to infinity if we bound the permissible $\pi$-value to not exceed the margin, $\Delta p$ by which the agents are on average more likely to vote for the correct alternative than for any competitor such that $\pi \leq \sfrac{\Delta p}{\Delta p + 1}$. 

\begin{theorem}\label{thm:bounds} 
Consider an approval voting setting with $m > 1$ alternatives, satisfying private agent approval independence (\Cref{def:independence}) and $\Delta p$-group reliability (\Cref{def:reliability})  for some $\Delta p \in (0,1]$, influenced by an opinion leader with    
$\pi \in [0,\sfrac{\Delta p}{\Delta p + 1})$ and $\hat{p} \in [0,1]$.
Then, it is guaranteed that the success probability, $P_\mathrm{min}$, of the approval voting process is at least
\begin{align}
1 - (m-1) \Big(\hat{p}  e^{-\frac{1}{2}n\Delta p^2(1-\pi)^2}  + (1 - \hat{p}) e^{-\frac{1}{2}n (\Delta p (1 {-} \pi) - \pi)^2} \Big)  . \nonumber
\end{align}
\end{theorem}

The following example illustrates a combination of the underlying parameters that results in a $P_\mathrm{min}$-value.


\begin{example}
Consider a set of $n = 100$ agents who are relatively competent, $\Delta p = 0.3$, and vote on a set of $m = 5$ alternatives. Suppose there is also a moderately influential opinion leader, $\pi = 0.1$, who is relatively likely to approve the correct alternative, $\hat{p} = 0.5$. This gives a worst case success probability of $P_{min} = 0.48$ for the correct alternative to win the approval vote.
\end{example}

Having outlined the voting and probabilistic framework as well as stated a jury theorem result based on the classical OL model for dependence, we now provide translations from the OL model to standard correlation metrics.



\section{Translating Measures of Correlation}

In this section, we present explicit translations from the probability of following the opinion leader ($\pi$) to two widely used dependence measures: covariance and the correlation coefficient. As a foundational step, we begin with a common special case aligned with several Condorcet Jury Theorem frameworks: we assume homogeneous competence levels for all agents and for the opinion leader. This matches the original OL model of Boland et al.\ \cite{Boland}. While related frameworks such as Everaere et al.\  also use homogeneous competencies in approval voting \cite{everaere2010epistemic}, they typically assume voter independence, a condition rarely met in practice. By incorporating an OL while maintaining homogeneous competence, we can leverage the bound of Karge et al.\ \cite{OL}, which subsumes the assumptions needed for our translations.

By extending these and similar settings to include an equally competent opinion leader and explicitly weakening the independence assumption, our work enables the consideration of applications that not only involve dependence structures among agents but also require the quantification of these dependencies using standard correlation measures.


Concretely, in this initial homogeneous case, the competence for approving any given alternative $\omega_j$ is the same for all agents and the OL: $p^{\omega_j}_i = \hat p^{\omega_j} = p^{\omega_j}$ for all $i$ and $j$. For the relevant alternative $\omega_j$, we write this common value simply as $p$.

\paragraph{Standing assumptions for the translations.}
We fix the true alternative $\omega_*=\omega_j$ and work throughout conditional on $\omega_*$ (we omit the conditioning bar).
We assume homogeneous competence on the true state, i.e., $\mathbb{P}(X_i^{\omega_*}=1\mid \omega_*)=\mathbb{P}(X_o^{\omega_*}=1\mid \omega_*)=p$ for all $i$, and conditional independence of private approvals across agents and the OL given $\omega_*$. The final vote is $V_i^{\omega_j}=X_o^{\omega_j}$ with probability $\pi$ and $V_i^{\omega_j}=X_i^{\omega_j}$ otherwise, with the following indicators independent across agents and of all signals.

To begin our derivation, we first establish the probabilities of an agent's final vote, $V_i^{\omega_j}$, conditioned on the opinion leader's choice, $X_o^{\omega_j}$. These probabilities are based on the OL model \cite{Boland}.

First, consider the case where the OL approves alternative $\omega_j$ (i.e., $X_o^{\omega_j} = 1$):
\begin{align}
\mathbb{P}(V_i^{\omega_j} = 1 \mid X_o^{\omega_j} = 1) &= p + \pi (1-p) \label{eq:V1OL1} \\
\mathbb{P}(V_i^{\omega_j} = 0 \mid X_o^{\omega_j} = 1) &= (1-p) - \pi (1-p) \label{eq:V0OL1}
\end{align}
These expressions represent the probability that agent $i$'s final vote is for ($V_i^{\omega_j} = 1$) or against ($V_i^{\omega_j} = 0$) alternative $\omega_j$, given that the opinion leader approves it.
Conversely, consider the case where the OL disapproves alternative $\omega_j$ (i.e., $X_o^{\omega_j} = 0$):
\begin{align}
\mathbb{P}(V_i^{\omega_j} = 1 \mid X_o^{\omega_j} = 0) &= p - \pi p \label{eq:V1OL0} \\
\mathbb{P}(V_i^{\omega_j} = 0 \mid X_o^{\omega_j} = 0) &= (1-p) + \pi p \label{eq:V0OL0}
\end{align}

\subsection{Derivation for the Homogeneous Case}

We now derive the covariance and correlation coefficient for the homogeneous case, where all agents and the opinion leader share the same competence $p$.

\paragraph{Expectation and Variance.}
For a binary random variable, its expectation is equal to the probability of it taking the value 1. We first determine the overall probability that an agent's final vote is 1, using the law of total probability:
\begin{align}
\mathbb{P}(V_i^{\omega_j} = 1) &= \mathbb{P}(V_i^{\omega_j} = 1 \mid X_o^{\omega_j}=1)\mathbb{P}(X_o^{\omega_j}=1) \nonumber \\
& \quad + \mathbb{P}(V_i^{\omega_j} = 1 \mid X_o^{\omega_j}=0)\mathbb{P}(X_o^{\omega_j}=0) \nonumber \\
&= (p + \pi (1-p))p + (p - \pi p)(1-p) \nonumber \\
&= p. \label{eq:PV1}
\end{align}



Hence $\mathbb{E}[V_i^{\omega_j}]=p$ and, since $V_i^{\omega_j}$ is Bernoulli,
\begin{equation}
\mathrm{Var}(V_i^{\omega_j}) = p(1-p). \label{eq:VarV_homo}
\end{equation}

\paragraph{Covariance.}
To compute the covariance between the votes of two distinct agents, $V_i^{\omega_j}$ and $V_k^{\omega_j}$ ($i \neq k$), we first need to compute the expectation of their product, $\mathbb{E}(V_i^{\omega_j}V_k^{\omega_j})$. 


Since $V_i^{\omega_j}$ and $V_k^{\omega_j}$ are conditionally independent given the OL's choice, we apply the law of total probability:
\begin{align}
&\mathbb{P}(V_i^{\omega_j}=1, V_k^{\omega_j}=1) \nonumber \\
&= \mathbb{P}(V_i^{\omega_j}=1, V_k^{\omega_j}=1 \mid X_o^{\omega_j}=1)\mathbb{P}(X_o^{\omega_j}=1) \nonumber \\
&\quad + \mathbb{P}(V_i^{\omega_j}=1, V_k^{\omega_j}=1 \mid X_o^{\omega_j}=0)\mathbb{P}(X_o^{\omega_j}=0) \nonumber \\
&= (\mathbb{P}(V_i^{\omega_j}=1 \mid X_o^{\omega_j}=1) \mathbb{P}(V_k^{\omega_j}=1 \mid X_o^{\omega_j}=1))p \nonumber \\
&\quad + (\mathbb{P}(V_i^{\omega_j}=1 \mid X_o^{\omega_j}=0) \mathbb{P}(V_k^{\omega_j}=1 \mid X_o^{\omega_j}=0))(1-p) \nonumber \\
&= (p + \pi (1-p))^2 p + (p - \pi p)^2 (1-p) \nonumber \\
&= p^2 + \pi^2 p(1-p). \label{eq:expproduct_homo}
\end{align}
Using $\text{Cov}(V_i^{\omega_j}, V_k^{\omega_j}) = \mathbb{E}(V_i^{\omega_j}V_k^{\omega_j}) - \mathbb{E}(V_i^{\omega_j})\mathbb{E}(V_k^{\omega_j})$:
\begin{align}
\text{Cov}(V_i^{\omega_j}, V_k^{\omega_j}) &= (p^2 + \pi^2 p(1-p)) - p \cdot p \nonumber \\
&= \pi^2 p(1-p). \label{eq:Covariance_homo}
\end{align}

This not only provides the direct translation from covariance into the OL model, i.e.\
$\text{Cov}(V_i^{\omega_j}, V_k^{\omega_j}) = \pi^2\, p(1-p)$,
but also the reverse direction by solving for $\pi$:
\[
\pi \;=\; \sqrt{\sfrac{\text{Cov}(V_i^{\omega_j}, V_k^{\omega_j})}{p(1-p)}}\,.
\]

\paragraph{Correlation Coefficient.}

The correlation coefficient, $\rho$, standardizes the covariance by dividing it by the product of the standard deviations. For the homogeneous case, the standard deviations are equal, $\sqrt{\text{Var}(V_i^{\omega_j})} = \sqrt{p(1-p)}$. Thus, we obtain:

\begin{align}
\rho &= \frac{\text{Cov}(V_i^{\omega_j},V_k^{\omega_j})}{\sqrt{\text{Var}(V_i^{\omega_j})\text{Var}(V_k^{\omega_j})}} \nonumber \\
&= \frac{\pi^2 p(1-p)}{\sqrt{(p(1-p))(p(1-p))}} \nonumber \\
&= \frac{\pi^2 p(1-p)}{p(1-p)} = \pi^2. \label{eq:rho_homo}
\end{align}
Thus, in the homogeneous case with $\hat p=p$, we obtain the direct translation $\rho=\pi^2$ and $\pi=\sqrt{\rho}$.

\subsection{Translated Jury Theorem}

Building upon our translations, we now apply these relationships to reformulate the jury theorem's guarantees in terms of our derived correlation measures. This approach yields new bounds on the success probability that are expressed using standard statistical metrics rather than the abstract opinion leader influence parameter.

To derive the bounds on the correlation coefficient ($\rho$) and covariance ($\mathrm{Cov}$) that maintain the asymptotic success of the jury, we begin with the threshold from Theorem \ref{thm:bounds} on the permissible OL influence
($\pi < \frac{\Delta p}{\Delta p + 1}$).

To obtain equivalent bounds for $\rho$ and $\mathrm{Cov}$, we  use the direct translations derived in the previous section, and substitute accordingly for $\pi$:

\begin{center}
\begin{minipage}[t]{0.48\columnwidth}
\textit{Correlation Coefficient:}
\begin{align*}
\sqrt{\rho} &< \frac{\Delta p}{\Delta p + 1} \\
\rho &< \frac{\Delta p^2}{(\Delta p + 1)^2}.
\end{align*}
\end{minipage}
\hfill
\begin{minipage}[t]{0.48\columnwidth}
\textit{Covariance:}
\begin{align*}
\sqrt{\frac{\mathrm{Cov}}{p(1-p)}} &< \frac{\Delta p}{\Delta p + 1} \\
\mathrm{Cov} &< \frac{\Delta p^2 p(1-p)}{(\Delta p + 1)^2}. 
\end{align*}
\end{minipage}
\end{center}

These derived bounds ensure that the asymptotic convergence of the jury's success probability to 1 is preserved under the new metrics.

\begin{corollary}\label{cor:translated_bounds}
Consider the setting of Theorem~\ref{thm:bounds} with homogeneous competence on the true state $p$ for all agents and the OL (i.e., $p_i=\hat p=p$). 
If either
\[
\rho \in \Big[0,\;\frac{\Delta p^2}{(\Delta p + 1)^2}\Big)
\quad\text{or}\quad
\operatorname{Cov} \in \Big[0,\;\frac{\Delta p^2\,p(1-p)}{(\Delta p + 1)^2}\Big),
\]
then the success probability satisfies
\begin{align*}
P_{\min}\;\ge\; 1 - (m-1) \Big( & p \, e^{-\frac{1}{2}n\Delta p^2(1-\sqrt{\rho})^2} \\
&+ (1 - p) \, e^{-\frac{1}{2}n \big(\Delta p (1 {-} \sqrt{\rho}) - \sqrt{\rho}\big)^2} \Big),
\end{align*}
and, equivalently, in covariance form,
\begin{align*}
P_{\min}\;\ge\; 1 - (m-1) \Big(& p \, e^{-\frac{1}{2}n\Delta p^2\!\left(1-\frac{\sqrt{\operatorname{Cov}}}{\sqrt{p(1-p)}}\right)^{\!2}}\\
&+ (1 - p)\, e^{-\frac{1}{2}n \left(\Delta p \!\left(1 {-} \frac{\sqrt{\operatorname{Cov}}}{\sqrt{p(1-p)}}\right) - \frac{\sqrt{\operatorname{Cov}}}{\sqrt{p(1-p)}}\right)^{\!2}} \Big).
\end{align*}
In particular, $P_{\min}\to 1$ as $n\to\infty$ under either condition.
\end{corollary}

Convergence holds as we apply Theorem~\ref{thm:bounds} (which was shown to converge in \cite{OL})
with $\hat p=p$ and substitute $\pi=\sqrt{\rho}$ (or $\pi=\sqrt{\sfrac{\operatorname{Cov}}{p(1-p)}}$).
The threshold $\pi<\sfrac{\Delta p}{\Delta p+1}$ is then equivalent to 
$\rho<\sfrac{\Delta p^2}{(\Delta p + 1)^2}$ and 
$\operatorname{Cov}<\sfrac{\Delta p^2\,p(1-p)}{(\Delta p + 1)^2}$.

To conclude this section, we illustrate the translated values for covariance and the correlation coefficient in terms of $\pi$ and a prescribed competency value $p$ in Figure \ref{doubleplot}.

\begin{figure*}[t] 
    \centering 
    \begin{minipage}[b]{0.48\linewidth} 
        \centering
 \begin{tikzpicture}
            \begin{axis}[
                width=\linewidth, 
                height=0.8\linewidth, 
                xlabel={$\pi$},
                ylabel={Translated Correlation Value},
                legend style={at={(0.05,0.95)},anchor=north west},
                grid=major,
                ymin=0,
                xmin=0,
                xmax=1,
                ymax=1,
            ]
            \addplot[
            domain=0:1, 
            samples=50, 
            color=blue, 
            thick 
        ]
        {x^2*0.4*(1-0.4)};
        \addlegendentry{Covariance};

        \addplot[
            domain=0:1, 
            samples=50, 
            color=red, 
            thick 
        ]
        {x^2};
        \addlegendentry{Correlation Coefficient};

            \end{axis}
        \end{tikzpicture}

    \end{minipage}%
    \hfill 
    \begin{minipage}[b]{0.48\linewidth}
        \centering
        \begin{tikzpicture}
            \begin{axis}[
                width=\linewidth, 
                height=0.8\linewidth, 
                xlabel={$\pi$},
                ylabel={$p$},
                zlabel={Cov},
                zmin=0,
                zmax=0.25,
                xmin=0,
                xmax=1,
                ymin=0,
                ymax=1,
                ztick={0.1,0.2},
                view={20}{30},
                grid=major,
                colormap/jet,
            ]
              \addplot3[
                surf,
                domain=0:1, 
                domain y=0:1, 
                samples=50, 
                samples y=50, 
                ]
            { 
               x^2*y*(1-y)};
            \addlegendentry{Covariance};
            \end{axis}
        \end{tikzpicture}

    \end{minipage}
    \caption{Translated correlation measures depending on $\pi$ for fixed $p = 0.4$ (left) and covariance depending on $\pi$ and $p$ (right).} 
\label{doubleplot}
\end{figure*}
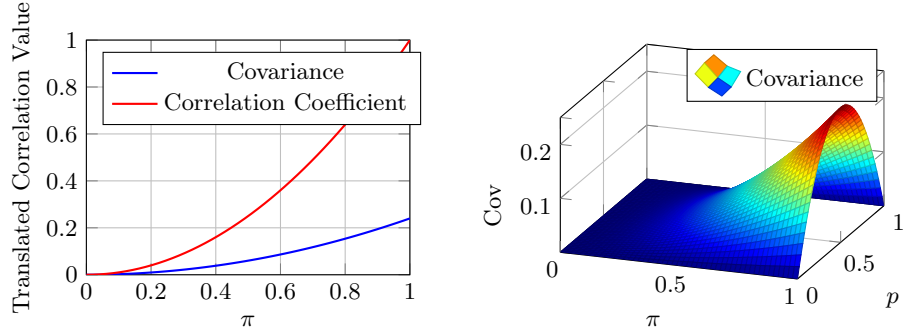

\subsection{Extension to the Heterogeneous Case}

We now extend the analysis to allow for heterogeneous competence levels. Generally, there are two steps of heterogeneity that we can consider. First, we consider the case where the agents remain homogeneous, but the OL's competency differs. Second, we discuss the case where we have full heterogeneity across agents and the OL. 

\paragraph{Homogeneous agents with $\hat p\neq p$}
If agents are homogeneous ($p_i\equiv p$) but the OL competence may differ ($\hat p\neq p$), then
$$
\mathbb{E}({V_i}) \;=\; \mu \;=\; \pi \hat p + (1-\pi)p,
\qquad
\Cov{V_i}{V_k} \;=\; \pi^2 \hat p(1-\hat p).
$$
From this, we obtain
$$
\rho \;=\; \frac{\pi^2 \hat p(1-\hat p)}{\mu(1-\mu)}.
$$
These identities follow by conditioning on the OL’s choice $X_o$ exactly as in the homogeneous case:
the expectation is a convex combination of $\hat p$ and $p$, the covariance reduces to $\pi^2\mathrm{Var}(X_o)$,
and the correlation coefficient is obtained by normalizing with the common variance $\mu(1-\mu)$. 
Thus $\rho=\pi^2$ holds if $\hat p=p$. The translated bounds remain valid upon replacing
$p(1-p)$ with $\hat p(1-\hat p)$ and $p$ in the exponents’ weights by $\hat p$. The full derivation can be found in the appendix.

\paragraph{Heterogeneous agents and OL} Let $p_i$ denote the private competence level of agent $i$, meaning $\mathbb{P}(X^{\omega_j}_{i}=1) = p_i$. The opinion leader's competence remains $\hat{p}$, and the influence parameter is $\pi$.

The probability that agent $i$'s final vote is 1 is:
\begin{align*}
\mathbb{P}(V_i^{\omega_j} = 1) &= (p_i + \pi (1-p_i))\hat{p} + (p_i - \pi p_i)(1-\hat{p}) \\
&= p_i + \pi(\hat{p} - p_i)
\end{align*}
Thus, $\mathbb{E}(V_i^{\omega_j}) = p_i + \pi(\hat{p} - p_i)$. The variance of agent $i$'s vote is derived from this expectation:
$$ \text{Var}(V_i^{\omega_j}) = (p_i + \pi(\hat{p} - p_i)) (1 - (p_i + \pi(\hat{p} - p_i))) $$
Note that each agent now has a unique variance for their vote, depending on their $p_i$.

\paragraph{Covariance.}
For any two distinct agents, $i$ and $k$ ($i \neq k$), their votes are conditionally independent given the OL's choice. We can compute the covariance using the law of total expectation and the property of conditional covariance, which states with respect to $X_o^{\omega}$,
$
\operatorname{Cov}\!\left(V_i^{\omega},V_k^{\omega}\right)
= \mathbb{E}\!\left[\operatorname{Cov}\!\left(V_i^{\omega},V_k^{\omega}\mid X_o^{\omega}\right)\right]
\;+\; \operatorname{Cov}\!\left(\mathbb{E}\!\left[V_i^{\omega}\mid X_o^{\omega}\right],\,\mathbb{E}\!\left[V_k^{\omega}\mid X_o^{\omega}\right]\right).
$
Since agent votes are conditionally independent given the OL's choice ($X_o^{\omega_j}$), the first term ($\mathbb{E}\!\left[\operatorname{Cov}\!\left(V_i^{\omega},V_k^{\omega}\mid X_o^{\omega}\right)\right]$) is zero as $V_i^{\omega_j}$ and $V_i^{\omega_k}$ have no covariance once the OL's choice is fixed. We therefore only need to compute the second term.

First, we determine the conditional expectations for each agent's vote given the OL's choice, $X_o^{\omega_j}$:
\begin{align*}
\mathbb{E}[V_i^{\omega_j} \mid X_o^{\omega_j}] &= \mathbb{P}(V_i^{\omega_j}=1 \mid X_o^{\omega_j}) \\
&= (p_i + \pi (1-p_i))X_o^{\omega_j} + (p_i - \pi p_i)(1-X_o^{\omega_j}) \\
&= p_i - \pi p_i + X_o^{\omega_j}(p_i + \pi(1-p_i) - (p_i-\pi p_i)) \\
&= p_i - \pi p_i + X_o^{\omega_j}(\pi - \pi p_i + \pi p_i) \\
&= p_i + \pi(X_o^{\omega_j} - p_i)
\end{align*}
Similarly, for agent $k$:
$$
\mathbb{E}[V_k^{\omega_j} \mid X_o^{\omega_j}] = p_k + \pi(X_o^{\omega_j} - p_k)
$$
Now, we can compute the covariance between these two conditional expectations:
\begin{align*}
\text{Cov}(V_i^{\omega_j}, V_k^{\omega_j}) &= \text{Cov}(\mathbb{E}[V_i^{\omega_j} \mid X_o^{\omega_j}], \mathbb{E}[V_k^{\omega_j} \mid X_o^{\omega_j}]) \\
&= \text{Cov}(p_i + \pi(X_o^{\omega_j} - p_i), p_k + \pi(X_o^{\omega_j} - p_k)) \\
&= \text{Cov}(\pi X_o^{\omega_j}, \pi X_o^{\omega_j}) \\
&= \pi^2 \text{Var}(X_o^{\omega_j}) \\
&= \pi^2 \hat{p}(1-\hat{p})
\end{align*}
Remarkably, the covariance between any two agents $i$ and $k$ is solely determined by the opinion leader's competence, $\hat{p}$, and the influence parameter, $\pi$, irrespective of the individual agent competencies, $p_i$ and $p_k$. This result provides a bidirectional translation for the heterogeneous case: $ \pi = \sfrac{\sqrt{Cov}}{\sqrt{\hat{p}-(\hat{p})^2}}$.

Since covariance is uniform across all pairs, the success probability bounds in Corollary \ref{cor:translated_bounds} directly extend to the heterogeneous case when expressed in terms of covariance. Thus, the same corollary as above applies.

\paragraph{Correlation Coefficient.}
The correlation coefficient between agents $i$ and $k$, $\rho_{ik}$, is given by:
$$ \rho_{ik} = \frac{\text{Cov}(V_i^{\omega_j},V_k^{\omega_j})}{\sqrt{\text{Var}(V_i^{\omega_j})\text{Var}(V_k^{\omega_j})}} $$

Substituting the derived covariance and variances, we get:

\begin{align}
\rho_{ik} &= \frac{\pi^2 \hat{p}(1 - \hat{p})}
{\sqrt{(p_i + \pi(\hat{p} - p_i))(1 - (p_i + \pi(\hat{p} - p_i)))}} \nonumber \\
&\quad \times \frac{1}
{\sqrt{(p_k + \pi(\hat{p} - p_k))(1 - (p_k + \pi(\hat{p} - p_k)))}} \label{eq:rho_hetero}
\end{align}

This expression shows that the correlation coefficient now explicitly depends on the individual competencies $p_i$ and $p_k$ of the two agents, meaning that $\rho_{ik}$ will generally be unique for each pair of agents.


That is, unlike the covariance translation in the heterogeneous case, we cannot derive a single, general translation for the correlation coefficient in the heterogeneous case since the correlation coefficient is not uniform across all pairs of agents. Thereby, we can moreover not derive a novel CJT-bound for the heterogeneous correlation coefficient straightforwardly. Intuitively, this reflects the fact that the strength of correlation between agents depends not only on the leader’s influence $\pi$ but also on how well each individual agent performs relative to the leader.

Since there is no single value of $\rho$ that describes the entire group, it is impossible to define a universal threshold for the group correlation that guarantees the success of the jury theorem. Instead we derive a bound akin to Corollary \ref{cor:translated_bounds} where we bound the maximal $\rho$-value. 

To derive this bound, we require that each agent’s vote exhibits genuine randomness. That is, we assume that

$$\sigma_{\min} \le \mathrm{Var}(V_i^{\omega_j}) \le \sigma_{\max}.$$ 

With that, $\sigma_{\min}>0$ excludes the degenerate cases where an agent’s vote is deterministic (zero variance), while $\sigma_{\max}\le 1/4$ reflects the natural upper bound for the variance of a Bernoulli random variable.

From the bounded variance, we can define the maximal $\rho$-value, and bound this in terms of $\pi$. That is, for each pair of distinct agents $(i,k)$, let $\rho_{ik}$ denote their correlation coefficient, and define the worst-case pairwise correlation
$$
\rho_{\max} \ :=\ \max_{i\neq k}\ \rho_{ik}.
$$

\paragraph{Bounding OL influence in terms of $\pi$.} 


We derive the bound on $\pi$ as follows. Recall from our earlier derivation, the covariance is
$$
\mathrm{Cov}(V_i^{\omega_j},V_k^{\omega_j}) = \pi^2\hat p(1-\hat p).
$$
 Using the variance bounds and the definition of $\rho_{ik}$, we obtain
$$
\frac{\pi^2 \hat p(1-\hat p)}{\sigma_{\max}} \ \le\ \rho_{ik} \ \le\ \frac{\pi^2 \hat p(1-\hat p)}{\sigma_{\min}}.
$$
Taking maxima over all pairs $(i,k)$ yields
$$
\frac{\pi^2 \hat p(1-\hat p)}{\sigma_{\max}} \ \le\ \rho_{\max} \ \le\ \frac{\pi^2 \hat p(1-\hat p)}{\sigma_{\min}}.
$$

Rearranging gives the two-sided inequality
$$
\sqrt{\frac{\rho_{\max}\,\sigma_{\min}}{\hat p(1-\hat p)}} \ \le\ \pi \ \le\ \sqrt{\frac{\rho_{\max}\,\sigma_{\max}}{\hat p(1-\hat p)}}.
$$
Thus $\pi$ is bounded above by $\vartheta := \sqrt{\frac{\rho_{\max}\sigma_{\max}}{\hat p(1-\hat p)}}$.

\paragraph{CJT-bound threshold.} From this, we obtain a threshold on the permissible $\rho_{\max}$-value. Since

\[
\pi \le \sqrt{\frac{\rho_{\max}\,\sigma_{\max}}{\hat p(1-\hat p)}}\ \text{ and }\ 
\pi<\frac{\Delta p}{\Delta p+1}
\ \Rightarrow\
\sqrt{\frac{\rho_{\max}\,\sigma_{\max}}{\hat p(1-\hat p)}}<\frac{\Delta p}{\Delta p+1},
\]
we can solve for $\rho_{\max}$, and obtain

$$\rho_{\max}<\frac{\hat p(1-\hat p)}{\sigma_{\max}}\left(\frac{\Delta p}{\Delta p+1}\right)^{2}.$$

With that, we are ready to state a CJT-bound based on $\rho_{\max}$:

\begin{corollary}[CJT-Bound via $\rho_{\max}$]\label{cor:rho-max}
Consider the approval-voting setting of Theorem~\ref{thm:bounds} with heterogeneous private competencies $p_i$, opinion leader competence $\hat p\in[0,1]$, and influence $\pi\in[0,1]$. 
Assume there exist constants $0 < \sigma_{\min} \le \sigma_{\max} \le \tfrac{1}{4}$ such that the variance of each agent's final vote satisfies
$$
\sigma_{\min} \ \le\ \mathrm{Var}(V_i^{\omega_j}) \ \le\ \sigma_{\max}
\qquad \text{for all agents $i$.}
$$

In particular, if
\[
\rho_{\max} \;<\; \frac{\hat p(1-\hat p)}{\sigma_{\max}}\left(\frac{\Delta p}{\Delta p+1}\right)^{\!2},
\]
then $\pi<\frac{\Delta p}{\Delta p+1}$ and
\begin{align*}
P_{\min}\ \ge\ 1 - (m-1) \Big(& \hat{p}\, e^{\!\big(-\tfrac{1}{2} n \Delta p^2 (1-\vartheta)^2\big)}  + (1-\hat{p})\, e^{\!\big(-\tfrac{1}{2} n \big(\Delta p(1-\vartheta)-\vartheta\big)^2\big) }\Big),
\end{align*}
with $\vartheta := \sqrt{\frac{\rho_{\max}\sigma_{\max}}{\hat p(1-\hat p)}}$. Hence $P_{\min}\to 1$ as $n\to\infty$.
\end{corollary}

Recall that Theorem~\ref{thm:bounds} guarantees asymptotic success if 
$\pi < \sfrac{\Delta p}{\Delta p+1}$. 
From the variance-based translation we have 
$\pi \le \vartheta := \sqrt{\sfrac{\rho_{\max}\,\sigma_{\max}}{\hat p(1-\hat p)}}$. 
Hence, if $\vartheta < \sfrac{\Delta p}{\Delta p+1}$, then also 
$\pi < \sfrac{\Delta p}{\Delta p+1}$. 
Since the lower bound $P_{\min}(\pi)=1-(m-1)g(\pi)$ is \emph{strictly decreasing} in $\pi$ (Lemma~\ref{mono}), 
replacing $\pi$ by its upper bound yields the conservative inequality 
$P_{\min}(\pi) \ge P_{\min}(\vartheta)$. 
Moreover, because $\vartheta<\sfrac{\Delta p}{\Delta p+1}$, the exponents remain negative and 
$P_{\min}(\vartheta)\to 1$ as $n\to\infty$.

\begin{lemma}[Monotonicity of the bound in $\pi$]\label{mono}
Let
\[
g(\pi)\;:=\;
\hat{p}\, e^{-\frac{1}{2} n \Delta p^2 (1-\pi)^2}
\;+\;
(1-\hat{p})\, e^{-\frac{1}{2} n \big(\Delta p(1-\pi)-\pi\big)^2},
\]
so that $P_{\min}(\pi)=1-(m-1)\,g(\pi)$ in Theorem~\ref{thm:bounds}.
Then $g$ is strictly increasing on $\bigl[0,\sfrac{\Delta p}{\Delta p+1}\bigr)$ and hence
$P_{\min}$ is strictly decreasing on this interval. In particular, for any
$0\le \pi\le \vartheta<\sfrac{\Delta p}{\Delta p+1}$,
\[
P_{\min}(\pi)\;\ge\;P_{\min}(\vartheta).
\]
\end{lemma}

\begin{proof}
To prove monotonicity, we show that the first derivative of $g(\pi)$ is nonnegative. The full proof can be found in the appendix.
\end{proof}

\section{Summary}

In this paper, we take the first steps toward establishing a unified framework for modeling dependence within the Condorcet Jury Theorem (CJT). To this end, we present translations from the classical opinion leader model of voter dependence in CJT literature to two standard measures of dependence: covariance and correlation coefficients for homogeneous and heterogeneous agents, respectively. Additionally, we demonstrate how to use a previously derived bound on the worst-case probability of a voter identifying the correct alternative in a setting with an opinion leader to derive bounds for the newly translated metrics, preserving the asymptotic guarantees of the original CJT. 

\clearpage
\appendix
\section{Technical Appendix}


\subsection*{Proof for Diverging Opinion Leader Competency.}

\begin{proof}[Homogeneous agents and an OL with $\hat p \neq p$]
Let agents be homogeneous with private competence $p_i \equiv p \in [0,1]$ and let the opinion leader (OL) have competence $\hat p \in [0,1]$. Each agent follows the OL with probability $\pi \in [0,1]$ (independently across agents and independently of all approval signals). For a fixed alternative $\omega$, denote the final binary vote of agent $i$ by $V_i^{\omega}$ and the OL’s approval by $X_o^{\omega} \sim \mathrm{Bernoulli}(\hat p)$. Then, for all distinct $i \neq k$:
\[
\mathbb{E}\!\left[V_i^{\omega}\right] \;=\; \mu \;=\; \pi \hat p + (1-\pi)p, 
\qquad
\operatorname{Cov}\!\left(V_i^{\omega},V_k^{\omega}\right) \;=\; \pi^2 \hat p(1-\hat p),
\]
and
\[
\rho_{ik}^{\omega} \;:=\; 
\frac{\operatorname{Cov}\!\left(V_i^{\omega},V_k^{\omega}\right)}
     {\sqrt{\operatorname{Var}\!\left(V_i^{\omega}\right)\,\operatorname{Var}\!\left(V_k^{\omega}\right)}}
\;=\; \frac{\pi^2 \hat p(1-\hat p)}{\mu(1-\mu)}.
\]
In particular, $\rho_{ik}^{\omega} = \pi^2$ if and only if $\hat p = p$.

\medskip
\noindent\textit{Expectation.}
Let $X_i^{\omega} \sim \mathrm{Bernoulli}(p)$ denote agent $i$’s private approval. By the OL mechanism, $V_i^{\omega}$ equals $X_o^{\omega}$ with probability $\pi$ and equals $X_i^{\omega}$ with probability $1-\pi$. Using independence of $X_o^{\omega}$ and $X_i^{\omega}$,
\[
\mathbb{E}\!\left[V_i^{\omega}\right]
= \pi\,\mathbb{E}\!\left[X_o^{\omega}\right] + (1-\pi)\,\mathbb{E}\!\left[X_i^{\omega}\right]
= \pi \hat p + (1-\pi)p
= \mu.
\]

\medskip
\noindent\textit{Variance.}
Since $V_i^{\omega}\in\{0,1\}$, it is Bernoulli with mean $\mu$, hence
\[
\operatorname{Var}\!\left(V_i^{\omega}\right) = \mu(1-\mu).
\]

\medskip
\noindent\textbf{Covariance.}
Using the law of total covariance with respect to $X_o^{\omega}$,
\[
\operatorname{Cov}\!\left(V_i^{\omega},V_k^{\omega}\right)
= \mathbb{E}\!\left[\operatorname{Cov}\!\left(V_i^{\omega},V_k^{\omega}\mid X_o^{\omega}\right)\right]
\;+\; \operatorname{Cov}\!\left(\mathbb{E}\!\left[V_i^{\omega}\mid X_o^{\omega}\right],\,\mathbb{E}\!\left[V_k^{\omega}\mid X_o^{\omega}\right]\right).
\]
Conditional on $X_o^{\omega}$, agents act independently (private signals are independent and following the OL is independent across agents), so the first term is zero. For the second term,
\[
\mathbb{E}\!\left[V_i^{\omega} \mid X_o^{\omega}\right] \;=\; \pi\,X_o^{\omega} + (1-\pi)\,p,
\]
hence
\[
\operatorname{Cov}\!\left(\mathbb{E}\!\left[V_i^{\omega}\mid X_o^{\omega}\right],\,\mathbb{E}\!\left[V_k^{\omega}\mid X_o^{\omega}\right]\right)
= \operatorname{Var}\!\left(\pi\,X_o^{\omega} + (1-\pi)\,p\right)
= \pi^{2}\operatorname{Var}\!\left(X_o^{\omega}\right)
= \pi^{2}\hat p(1-\hat p).
\]

\medskip
\noindent\textbf{Correlation Coefficient.}
Since agents are homogeneous, $\operatorname{Var}\!\left(V_i^{\omega}\right)=\operatorname{Var}\!\left(V_k^{\omega}\right)=\mu(1-\mu)$. Therefore
\[
\rho_{ik}^{\omega}
= \frac{\operatorname{Cov}\!\left(V_i^{\omega},V_k^{\omega}\right)}{\mu(1-\mu)}
= \frac{\pi^{2}\hat p(1-\hat p)}{\mu(1-\mu)}.
\]
The identity $\rho_{ik}^{\omega}=\pi^{2}$ holds if and only if $\mu(1-\mu)=\hat p(1-\hat p)$; with $\mu=\pi \hat p+(1-\pi)p$, this is equivalent to $\hat p=p$.
\end{proof}

\subsection*{Proof of Monotonicity.}

Recall that we let
\[
g(\pi)\;:=\; 
\hat{p}\, e^{\!\big(-\tfrac{1}{2} n \Delta p^2 (1-\pi)^2\big)}  
\;+\; 
(1-\hat{p})\, e^{\!\big(-\tfrac{1}{2} n \big(\Delta p(1-\pi)-\pi\big)^2\big)},
\]
so that $P_{\min}(\pi)=1-(m-1)\,g(\pi)$ in Theorem 2.

We aim to show that $g$ is strictly increasing on $\bigl[0,\sfrac{\Delta p}{\Delta p+1}\bigr)$ and hence
$P_{\min}$ is strictly decreasing on this interval. In particular, for any
$0\le \pi\le \vartheta<\sfrac{\Delta p}{\Delta p+1}$,
\[
P_{\min}(\pi)\;\ge\;P_{\min}(\vartheta).
\]

\begin{proof}
We compute the derivative $g'(\pi)$ by differentiating each term.

\medskip\noindent\textbf{First exponential term.}
Consider
\[
g_1(\pi)\;:=\;\hat p\, e^{\!\Big(-\tfrac{1}{2}n\,\Delta p^{2}\,(1-\pi)^{2}\Big)}.
\]
Using the chain rule with the inner function $(1-\pi)^{2}$, we obtain
\begin{align*}
\frac{d}{d\pi}\,g_1(\pi)
&= \hat p\, e^{\!\Big(-\tfrac{1}{2}n\,\Delta p^{2}\,(1-\pi)^{2}\Big)}\cdot
\frac{d}{d\pi}\!\Big(-\tfrac{1}{2}n\,\Delta p^{2}\,(1-\pi)^{2}\Big)\\
&= \hat p\, e^{\!\Big(-\tfrac{1}{2}n\,\Delta p^{2}\,(1-\pi)^{2}\Big)}\cdot
\Big(-\tfrac{1}{2}n\,\Delta p^{2}\cdot 2(1-\pi)\cdot(-1)\Big)\\
&= \hat p\,n\,\Delta p^{2}\,(1-\pi)\,
e^{\!\Big(-\tfrac{1}{2}n\,\Delta p^{2}\,(1-\pi)^{2}\Big)}.
\end{align*}
Every factor here is nonnegative for $\pi\in[0,1)$, and in fact strictly positive:
$\hat p\ge 0$, $n>0$, $\Delta p^{2}>0$, $(1-\pi)>0$, and the exponential is $>0$.
Hence
\[
\frac{d}{d\pi}\,g_1(\pi)>0\quad\text{for all }\pi\in[0,1).
\]
In particular, this holds on $\bigl[0,\sfrac{\Delta p}{\Delta p+1}\bigr)$.

\medskip\noindent\textbf{Second exponential term.}
Let
\[
g_2(\pi)\;:=\;(1-\hat p)\,e^{\!\Big(-\tfrac{1}{2}n\,\big(\Delta p(1-\pi)-\pi\big)^{2}\Big)}.
\]
Define
\[
f(\pi)\;:=\;\Delta p(1-\pi)-\pi\;=\;\Delta p-(\Delta p+1)\pi.
\]
By the chain rule (outer function $\exp(\,\cdot\,)$, inner function $-\tfrac{1}{2}n\,f(\pi)^{2}$),
\begin{align*}
\frac{d}{d\pi}\,g_2(\pi)
&=(1-\hat p)\,e^{\!\Big(-\tfrac{1}{2}n\,f(\pi)^{2}\Big)}\cdot
\frac{d}{d\pi}\!\Big(-\tfrac{1}{2}n\,f(\pi)^{2}\Big)\\
&=(1-\hat p)\,e^{\!\Big(-\tfrac{1}{2}n\,f(\pi)^{2}\Big)}\cdot
\Big(-\tfrac{1}{2}n\cdot 2 f(\pi)\cdot f'(\pi)\Big)\\
&=(1-\hat p)\,e^{\!\Big(-\tfrac{1}{2}n\,f(\pi)^{2}\Big)}\cdot\big(-n f(\pi) f'(\pi)\big).
\end{align*}
Compute $f'(\pi)$ explicitly from $f(\pi)=\Delta p-(\Delta p+1)\pi$:
\[
f'(\pi)=-(\Delta p+1).
\]
Therefore,
\begin{align*}
\frac{d}{d\pi}\,g_2(\pi)
&=(1-\hat p)\,e^{\!\Big(-\tfrac{1}{2}n\,f(\pi)^{2}\Big)}\cdot n(\Delta p+1)\,f(\pi)\\
&=(1-\hat p)\,n\,(\Delta p+1)\,\big(\Delta p(1-\pi)-\pi\big)\,
e^{\!\Big(-\tfrac{1}{2}n\,\big(\Delta p(1-\pi)-\pi\big)^{2}\Big)}.
\end{align*}
On $\pi\in\bigl[0,\sfrac{\Delta p}{\Delta p+1}\bigr)$ we have
\[
\Delta p(1-\pi)-\pi=\Delta p-(\Delta p+1)\pi>0,
\]
since $\pi<\sfrac{\Delta p}{\Delta p+1}\iff (\Delta p+1)\pi<\Delta p$.
All other factors are nonnegative with at least one strictly positive factor:
$(1-\hat p)\ge 0$, $n>0$, $(\Delta p+1)>0$, the exponential is $>0$, and
$\Delta p(1-\pi)-\pi>0$ on the open interval.
Hence
\[
\frac{d}{d\pi}\,g_2(\pi)>0\quad\text{for all }\pi\in\Bigl[0,\sfrac{\Delta p}{\Delta p+1}\Bigr).
\]

\medskip\noindent\textbf{Combine derivatives.}
Since $g(\pi)=g_1(\pi)+g_2(\pi)$ and each derivative is nonnegative on
$\bigl[0,\sfrac{\Delta p}{\Delta p+1}\bigr)$ with at least one strictly positive everywhere on this interval, we conclude
\[
g'(\pi)>0\quad\text{for all }\pi\in\Bigl[0,\sfrac{\Delta p}{\Delta p+1}\Bigr).
\]
Therefore $g$ is strictly increasing on $\bigl[0,\sfrac{\Delta p}{\Delta p+1}\bigr)$, and
\[
P_{\min}(\pi)\;=\;1-(m-1)g(\pi)
\]
is strictly decreasing on the same interval. The monotonicity inequality
$P_{\min}(\pi)\ge P_{\min}(\vartheta)$ for $0\le \pi\le \vartheta<\sfrac{\Delta p}{\Delta p+1}$ follows immediately.
\end{proof}

\begin{credits}
\subsubsection{\discintname}
The authors have no competing interests to declare that are
relevant to the content of this article.
\end{credits}

\bibliographystyle{splncs04}
\bibliography{mybibliography}

\end{document}